\pdfoutput=1
\RequirePackage{ifpdf}
\ifpdf 
\documentclass[pdftex]{sigma}
\else
\documentclass{sigma}
\fi

\usepackage{bbold}
\usepackage{bbm}

\numberwithin{equation}{section}

\newtheorem{Theorem}{Theorem}[section]
\newtheorem*{Theorem*}{Theorem}

\newtheorem{Proposition}[Theorem]{Proposition}
 { \theoremstyle{definition}

\newtheorem{Example}[Theorem]{Example}
\newtheorem{Remark}[Theorem]{Remark} }

\begin{document}
\allowdisplaybreaks

\newcommand{\arXivNumber}{2306.10609}

\renewcommand{\PaperNumber}{065}

\FirstPageHeading

\ShortArticleName{Realizations of the Extended Snyder Model}

\ArticleName{Realizations of the Extended Snyder Model}

\Author{Tea MARTINI\'C BILA\'C~$^{\rm a}$ and Stjepan MELJANAC~$^{\rm b}$}

\AuthorNameForHeading{T.~Martini\'c Bila\'c and S.~Meljanac}

\Address{$^{\rm a)}$~Faculty of Science, University of Split, Rudera Bo\v{s}kovi\'ca 33, 21000 Split, Croatia}
\EmailD{\href{mailto:teamar@pmfst.hr}{teamar@pmfst.hr}}

\Address{$^{\rm b)}$~Division of Theoretical Physics, Ruder Bo\v skovi\'c Institute, Bijeni\v cka cesta 54,\\
\hphantom{$^{\rm b)}$}~10002 Zagreb, Croatia}
\EmailD{\href{mailto:meljanac@irb.hr}{meljanac@irb.hr}}

\ArticleDates{Received June 21, 2023, in final form August 31, 2023; Published online September 14, 2023}

\Abstract{We present the exact realization of the extended Snyder model. Using similarity transformations, we construct realizations of the original Snyder and the extended Snyder models. Finally, we present the exact new realization of the $\kappa$-deformed extended Snyder model.}	


\Classification{81R60}

\section{Introduction}

The first example of NC geometry was presented in \cite{Snyder}. Fundamental length scale could be identified in natural way with Planck length $L_p= \sqrt{G\hbar /c^3}
\approx 1.62 \times 10^{-35}$~m \cite{Garay}. The length scale enters the theory through commutators of spacetime coordinates in~\cite{Amelino-2, Amelino-Camelia-1, Doplicher, Doplicher-2}. Deformations of spacetime symmetries-gravity, group-valued momenta, and noncommutative fields were presented in~\cite{Arzano}.

Coproduct and star product in the Snyder model were calculated in \cite{Battisti-1, Girelli} using ideas from development of NC geometry \cite{Majid}. However, in the Snyder model, the algebra generated by position operators is not closed and the bialgebra resulting from implementation of the coproduct is not a Hopf algebra. In particular, the coproduct is noncoassociative and the star product is nonassociative as well \cite{Battisti-1}.

A closed Lie algebra can be obtained if one adds generators of Lorentz algebra \cite{Girelli} to position generators. In this way one can define a Hopf algebra with a coassoaciative coproduct. If Lorentz generators are added as extended coordinates, we call this algebra extended Snyder algebra, and the theory based on this the extended Snyder model \cite{Meljanac-1}.

Some recent advances in the Snyder model are presented in \cite{Battisti-2, Battisti-1, Girelli, Meljanac-2}. Construction of field theory was addressed in \cite{Battisti-1,Franchino, Girelli} and different applications to phenomenology were considered in \cite{Mignemi-2, Mignemi-1}. Extensions in curved background were given in \cite{Banerjee, Guo, Kowalski, Lukierski-1, Meljanac-1, Meljanac-4, Meljanac-5, Mignemi-3}.

The Snyder model is defined as a Lie algebra generated by noncommutative coordinates $\hat{x}_{\mu}$ and Lorentz generators $M_{\mu\nu}$, $( M_{\mu\nu}=-M_{\nu\mu}) $, satisfying the commutation relations
\begin{align}
&\big[ \hat{x}_{\mu},\hat{x}_{\nu}\big] ={\rm i}\beta^{2}M_{\mu\nu},\qquad \mu,\nu=0,1,2,3,\qquad\beta \in \mathbb{R},\label{1.1}\\
&\big[ M_{\mu\nu},\hat{x}_{\lambda}\big] =-{\rm i}\big( \hat{x}_{\mu}\eta_{\nu\lambda}-
\hat{x}_{\nu}\eta_{\mu\lambda}\big), \label{1.2}\\
&[ M_{\mu\nu},M_{\rho\sigma}] ={\rm i}\big( \eta_{\mu\rho} M_{\nu\sigma} - \eta_{\mu\sigma} M_{\nu\rho} - \eta_{\nu\rho} M_{\mu\sigma} + \eta_{\nu\sigma} M_{\mu\rho}\big), \label{1.3}
\end{align}
where $\eta=\operatorname{diag}(-1,1,1,1)$ is the Minkowski metric.

Our goal is to construct realizations of the Snyder algebra \eqref{1.1}--\eqref{1.3} in terms of the Heisenberg algebra generated by
coordinates $x_{\mu}$ and momenta $p_{\mu}$ satisfying the commutation relations
\begin{gather*}
[ x_{\mu},{x_\nu}] =[ p_{\mu},p_{\nu}] =0, \qquad [ x_{\mu},p_{\nu}] ={\rm i}\eta_{\mu\nu}.
\end{gather*}

In Section \ref{sec-2}, we start with the original Snyder realization with $M_{\mu\nu}=x_{\mu}p_{\nu}-x_{\nu}p_{\mu}$ and use similarity transformations to construct a family of realizations of Snyder model. In Section~\ref{sec 3}, we apply this method to construct realizations of the extended Snyder model in which the Lorentz generators are realized by $M_{\mu\nu}=\hat{x}_{\mu\nu}+x_{\mu}p_{\nu}-x_{\nu}p_{\mu}$, where $\hat{x}_{\mu\nu}$ are additional tensorial generators. In Section \ref{sec-4}, we present the exact new realization of the $\kappa$-deformed extended Snyder model. Finally, in Section~\ref{sec-5}, we give the discussion and conclusion.

\section{Realizations of the Snyder model}\label{sec-2}
The original Snyder realization in terms of $x_{\mu}$ and $p_{\nu}$ is given by
\begin{align}
&\hat{x}_{\mu}=x_{\mu}+\beta^{2}(x\cdot p)p_{\mu},\label{2.1}\\
&M_{\mu\nu}=x_{\mu}p_{\nu}-x_{\nu}p_{\mu},\label{2.2}
\end{align}
where $x\cdot p=x_{\alpha}p_\alpha$.\footnote{We denote $x_{\alpha}p_{\alpha}=\sum_{\alpha,\beta=0}^{3} \eta_{\alpha\beta} x_{\alpha}p_{\alpha}$, and generally summation over pair of repeated indices is assumed.}
Further realizations of the Snyder model can be obtained by similarity transformations by the operator $S={\rm e}^{{\rm i}G}$, where
\begin{align}
&G=F_{0}(u)+(x\cdot p)F(u),\qquad u=\beta^{2}p^{2}, \qquad \beta \in \mathbb{R},\nonumber\\
&F_{0}(0)=0,\qquad F(0)=0,\qquad p^2=p_{\alpha}p_{\alpha}.\label{2.3}
\end{align}
Note that for $\beta^{2}=0$ we have $G=0$ and $S=\text{id}$
and $G$ is Lorentz invariant and linear in the coordinates $x_{\alpha}$.
\begin{Theorem}\label{1-T}
Using similarity transformation defined by $S={\rm e}^{{\rm i}G}$, where $G$ is given by \eqref{2.3}, we obtain the corresponding realizations of Snyder model
\begin{align*}
\hat{x}_{\mu}=S\big( x_{\mu}+\beta^{2}(x\cdot p)p_{\mu}\big) S^{-1}=x_{\mu}\varphi_{1}(u)+\beta^{2}(x\cdot p)p_{\mu}\varphi_{2}(u)+\beta^{2}p_{\mu}\varphi_{3}(u),
\end{align*}
where
\begin{equation}\label{2.4}
\varphi_{2}(u)=\frac{1+2\dot{\varphi}_{1}(u)\varphi_{1}(u)}{\varphi_{1}(u)-2u\dot{\varphi}_{1}(u)}, \qquad \dot{\varphi}_{1}=\frac{{\rm d}\varphi_{1}(u)}{{\rm d}u} \qquad \text{and} \qquad u=\beta^{2}p^{2}.
\end{equation}
\end{Theorem}

In order to prove the above theorem, first we prove the following propositions. Note that if~$F_{0}(u)=0$, then $\varphi_{3}(u)=0$, hence, for simplicity in what follows we assume that $F_{0}(u)=0$ and~$G=(x\cdot p)F(u)$.

\begin{Proposition}\label{1-P}
Let $x_{\mu}'=Sx_{\mu}S^{-1}$, where $S={\rm e}^{{\rm i}G}$ and $G=(x\cdot p)F(u)$. Then
\begin{equation}\label{2.5}
x_{\mu}'=x_{\mu}g_{1}(u)+\beta^{2}(x\cdot p)p_{\mu}g_{2}(u),
\end{equation}
where
\begin{equation}
g_{1}(u)=\big( {\rm e}^{F(1-2u\frac{\rm d}{{\rm d}u})}\big) (1).\label{2.10}
\end{equation}
\end{Proposition}
\begin{proof}
By defining the iterated commutator
\begin{equation*}
( {\rm ad}_{G}) ^{n}(x_{\mu})=\underset{n \qquad \text{times}}{\underbrace{[G,\ldots,[G,[G}},x_{\mu}]]\ldots], \qquad ( {\rm ad}_{G} ) ^{0}(x_{\mu})=x_{\mu},
\end{equation*}
and using the Hadamard formula, we have
\begin{equation}
x_{\mu}'=Sx_{\mu}S^{-1}={\rm e}^{{\rm i}G}x_{\mu}{\rm e}^{-{\rm i}G}=x_{\mu}+\sum_{n=1}^{\infty}\frac{( {\rm ad}_{{\rm i}G}) ^{n}(x_{\mu})}{n!}\label{2.8}.
\end{equation}
We prove relation \eqref{2.5} by induction on $n$. Using the Leibniz rule for adjoint representation and%
\begin{equation}\label{2.6}
[F,x_{\mu}]=-{\rm i}\frac{\partial F}{\partial p_{\mu}}=-{\rm i}2\beta^{2}p_{\mu}\dot{F}, \qquad F\equiv F(u) \qquad \text{and} \qquad \dot{F}=\frac{{\rm d}F}{{\rm d}u}
,\end{equation} it is easy to see that
for $n=1$ we have
\begin{equation*}
( {\rm ad}_{{\rm i}G}) (x_{\mu})={\rm i}[ (x\cdot p)F, x_{\mu}] =x_{\mu}g_{1 1}(u)+\beta^{2}(x\cdot p)p_{\mu}g_{2 1}(u),
\end{equation*}
where $g_{1 1}(u)=F$ and $g_{2 1}(u)=2\dot{F}$. In following, we denote $g_{ij}\equiv g_{ij}(u)$.
Assume that the relation
\begin{equation}\label{2.9}
( {\rm ad}_{{\rm i}G}) ^{n}(x_{\mu})=x_{\mu}g_{1 n}+\beta^{2}(x\cdot p)p_{\mu}g_{2 n}
\end{equation}
holds for some $n>1$. Then by the induction assumption, we have
\begin{align*}
&( {\rm ad}_{{\rm i}G}) ^{n+1}(x_{\mu})={\rm i}\big[ (x\cdot p)F, x_{\mu}g_{1 n}+\beta^{2}(x\cdot p)p_{\mu}g_{2 n}\big] =x_{\mu}g_{1( n+1) }+\beta^{2}(x\cdot p)p_{\mu}g_{2( n+1) },
\end{align*}
where, using the Leibniz rule and \eqref{2.6}, we obtain
\begin{equation}\label{2.12}
g_{1(n+1) }=Fg_{1 n}-2u\dot{g}_{1 n}F
\end{equation}
and
\begin{equation*}
g_{2(n+1) }=2\dot{F}g_{1 n}+2u\dot{F}g_{2 n}-g_{2 n}F-2u\dot{g}_{2 n}F.
\end{equation*}
Let us denote
\begin{equation*}
g_{1}(u)=\sum_{n=0}^{\infty}\frac{g_{1 n}}{n!},\qquad \text{where} \qquad g_{1 0}=1,
\end{equation*}
and
\begin{equation*}
g_{2}(u)=\sum_{n=1}^{\infty}\frac{g_{2 n}}{n!}.
\end{equation*}
Then, substituting \eqref{2.9} into \eqref{2.8}, it follows that \eqref{2.5} holds.
Now, we expand
\begin{equation*}
\big( {\rm e}^{F(1-2u\frac{\rm d}{{\rm d}u})}\big) (1)= \sum_{n=0}^{\infty}\frac{\big(F\big(1-2u\frac{\rm d}{{\rm d}u}\big)\big) ^{n}(1)}{n!}
\end{equation*}
and prove by induction on $n$ that
\begin{equation}\label{2.11}
g_{1 n}=\left(F\left(1-2u\frac{\rm d}{{\rm d}u}\right)\right)^{n} (1).
\end{equation}
Note that for $n=0$, we have $g_{1 0}={\rm id} (1)=1$ and
for $n=1$ it is easy to verify that
\begin{gather*}
\left(F\left(1-2u\frac{\rm d}{{\rm d}u}\right)\right) (1)=F=g_{11}.
\end{gather*}
Suppose that relation \eqref{2.11} is true for some $n>1$. By the induction assumption and from~\eqref{2.12}, we have
\begin{align*}
\left(F\left(1-2u\frac{\rm d}{{\rm d}u}\right)\right)^{n+1} (1)&=\left( \left(F\left(1-2u\frac{\rm d}{{\rm d}u}\right)\right) \circ \left(F\left(1-2u\frac{\rm d}{{\rm d}u}\right)\right) ^{n} \right) (1) \\
&=\left(F\left(1-2u\frac{\rm d}{{\rm d}u}\right)\right)( g_{1 n})= Fg_{1 n}-2u\dot{g}_{1 n}F=g_{1(n+1) },
\end{align*}
which proves our claim \eqref{2.11} and consequently \eqref{2.10} holds.
\end{proof}

\begin{Proposition}\label{2-P}
	Let $p_{\mu}'=Sp_{\mu}S^{-1}$, where $S={\rm e}^{{\rm i}G}$ and $G=(x\cdot p)F(u)$. Then
	\begin{equation}\label{2.14}
	p_{\mu}'=p_{\mu}g_{3}(u),
	\end{equation}
	where
	\begin{equation}
	g_{3}(u)=\big( {\rm e}^{-F(1+2u\frac{\rm d}{{\rm d}u})}\big) (1).\label{2.15}
	\end{equation}
\end{Proposition}
\begin{proof}
Analogous to the proof of the previous proposition, first by using the Hadamard formula, we find
\begin{equation*}
p_{\mu}'=Sp_{\mu}S^{-1}={\rm e}^{{\rm i}G}p_{\mu}{\rm e}^{-{\rm i}G}=p_{\mu}+\sum_{n=1}^{\infty}\frac{( {\rm ad}_{{\rm i}G}) ^{n}(p_{\mu})}{n!}\label{2.17}.
\end{equation*}
Then, by induction on $n$, we prove that
\begin{equation}\label{2.16}
( {\rm ad}_{{\rm i}G}) ^{n}(p_{\mu})=p_{\mu}g_{3 n}.
\end{equation}
After short computation, for $n=1$ we have
\begin{equation*}
{\rm i}[(x\cdot p)F, p_{\mu}]={\rm i}[(x\cdot p), p_{\mu}]F=p_{\mu}g_{31},
\end{equation*}
 where $g_{31}=-F$.
Assume that relation \eqref{2.16} holds for some $n>1$.
 Then by the induction assumption, we find
\begin{equation*}
( {\rm ad}_{{\rm i}G}) ^{n+1}(p_{\mu})={\rm i}[(x\cdot p)F,p_{\mu}g_{3 n}]=p_{\mu}g_{3 (n+1) },
\end{equation*}
where
\begin{equation}\label{2.19}
g_{3 (n+1)} =-Fg_{3 n}-2u\dot{g}_{3 n}F,
\end{equation}
which proves claim \eqref{2.16}. Finally, if we denote
\begin{equation*}
g_{3}(u)=\sum_{n=0}^{\infty}\frac{g_{3 n}}{n!},\qquad \text{where} \qquad g_{3 0}=1,
\end{equation*}
then \eqref{2.14} holds.
Also, we prove by induction on $n$ that
\begin{equation}\label{2.18}
g_{3 n}=\left(-F\left(1+2u\frac{\rm d}{{\rm d}u}\right)\right)^{n} (1).
\end{equation}
Note that for $n=0$, we have $g_{3 0}={\rm id}(1)=1$, and
for $n=1$, we get
\begin{equation*}
\left(-F\left(1+2u\frac{\rm d}{{\rm d}u}\right)\right) (1)=-F=g_{31}.
\end{equation*}
Suppose that relation \eqref{2.18} holds for some $n>1$. Then by the induction assumption and from~\eqref{2.19}, we have
\begin{align*}
\left(-F\left(1+2u\frac{\rm d}{{\rm d}u}\right)\right)^{n+1} (1)&=\left( \left(-F\left(1+2u\frac{\rm d}{{\rm d}u}\right)\right) \circ \left(-F\left(1+2u\frac{\rm d}{{\rm d}u}\right)\right) ^{n}\right) (1) \\
&=\left(-F\left(1+2u\frac{\rm d}{{\rm d}u}\right)\right)\left( g_{3 n}\right)\\
&= -Fg_{3 n}-2u\dot{g}_{3 n}F=g_{3 (n+1)}.
\end{align*}
Therefore, \eqref{2.18} holds for every $n$, which implies that \eqref{2.15} holds.
\end{proof}

Now, using results proven in the previous propositions, we can finally prove our main result given by Theorem~\ref{1-T}.

\begin{proof}[Proof of Theorem \ref{1-T}]
Let us denote $x_{\mu}'=Sx_{\mu}S^{-1}$ and $p_{\mu}'=Sp_{\mu}S^{-1}$, where $S={\rm e}^{{\rm i}G}$ and~$G=(x\cdot p)F(u)$. Then
\begin{equation}\label{2.20}
[x_{\mu}', x_{\nu}']=[p_{\mu}', p_{\nu}']=0, \qquad [x_{\mu}', p_{\nu}']={\rm i}\eta_{\mu\nu}
\end{equation}
and
\begin{equation}\label{2.21}
x_{\mu}'+\beta^{2} (x'\cdot p')p_{\mu}'=S\big( x_{\mu}+\beta^{2}(x\cdot p)p_{\mu}\big) S^{-1}.
\end{equation}
Inserting \eqref{2.5} and \eqref{2.14} into \eqref{2.20},
we get
\begin{align*}
&{\rm i}\eta_{\mu\nu}g_{1}g_{3}+{\rm i}p_{\nu}\frac{\partial g_{3}}{\partial p_{\mu}}g_{1}+{\rm i}\beta^{2}\left( \frac{\partial p_{\nu}}{\partial p_{\alpha}}p_{\alpha}g_{3}+\frac{\partial g_{3}}{\partial p_{\alpha}} p_{\alpha}p_{\nu}\right) p_{\mu}g_{2}={\rm i}\eta_{\mu\nu}, \\
& g_{i}\equiv g_{i}(u),\qquad i=1,2,3,
\end{align*}
which implies
\begin{equation}\label{2.22}
g_{3}=\frac{1}{g_{1}}
\end{equation}
and
\begin{equation}\label{2.23}
2g_{1}\dot{g}_{3}+g_{2}( g_{3}+2u\dot{g}_{3}) =0.
\end{equation}
Substituting \eqref{2.22} into \eqref{2.23},
we find
\begin{equation}\label{2.24}
g_{2}=\frac{2\dot{g}_{1}g_{1}}{g_{1}-2u\dot{g}_{1}}.
\end{equation}
Finally, using \eqref{2.5} and \eqref{2.14}, it follows from \eqref{2.21} that
\begin{equation}\label{2.25}
S\big( x_{\mu}+\beta^{2}(x\cdot p)p_{\mu}\big) S^{-1}=x_{\mu}g_{1}+\beta^{2}(x\cdot p)p_{\mu}\big( g_{2}+g_{3}+ug_{2}g_{3}^{2}\big).
\end{equation}
If we denote $\varphi_{1}=g_{1}$ and $\varphi_{2}=g_{2}+g_{3}+ug_{2}g^{2}_{3}$, then \eqref{2.4} follows from \eqref{2.22} and $\eqref{2.24}$.
\end{proof}
\begin{Example}
For $F_{0}=0$ and $F=-\frac{1}{2}u$, we get
\begin{equation*}
\varphi_{1}(u)=\sqrt{1-u}\qquad \text{and} \qquad \varphi_{2}(u)=0,
\end{equation*}
hence,
\begin{equation*}
\hat{x}_{\mu}=x_{\mu}\sqrt{1-u}.
\end{equation*}
\end{Example}
\begin{Remark}
If $F=0$ and $F_{0}\neq 0$, then ${x}_{\mu}'=x_{\mu}+2\beta^{2}p_{\mu}\dot{F}_{0}$, $p_{\mu}'=p_{\mu}$ and
\begin{equation*}
\hat{x}_{\mu}={\rm e}^{{\rm i}F_{0}}\big( x_{\mu}+\beta^{2}(x\cdot p)p_{\mu}\big) {\rm e}^{-{\rm i}F_{0}}=x_{\mu}+\beta^{2}(x\cdot p)p_{\mu}+2\beta^{2}p_{\mu}\dot{F}_{0}( 1+u).
\end{equation*}
\end{Remark}

\begin{Remark}
	When $\varphi_{1}(u)$ is fixed and $\varphi_{2}(u)$ is given with \eqref{2.4}, then $\varphi_{3}(u)$ depends on $F_{0}$ and can be arbitrary. There is a family of realizations with fixed $\varphi_{1}(u)$ and arbitrary $\varphi_{3}(u)$.
\end{Remark}
\begin{Remark}\label{rem-3}
A Hermitian realization can be obtained starting with the hermitian form of~\eqref{2.1}, that is
\begin{equation*}
\hat{x}_{\mu}=x_{\mu}+\frac{1}{2}\beta^{2}( (x\cdot p)p_{\mu}+p_{\mu}(p\cdot x))
\end{equation*}
and instead of $G$ writing $\frac{1}{2}\big( G+G^{\dagger}\big) $. Then result of Theorem~\ref{1-T} is obtained in hermitian form~$\frac{1}{2}\big( \hat{x}_{\mu}+\hat{x}_{\mu}^{\dagger}\big) $.
\end{Remark}

\section{Realizations of the extended Snyder model}\label{sec 3}

Different realizations of the Snyder algebra can be obtained introducing additional tensorial generators $\hat{x}_{\mu\nu}=-\hat{x}_{\nu\mu}$. This alternative approach was suggested in \cite{Girelli} and it was studied perturbatively from a different point of view in \cite{Meljanac-1, Meljanac-2, Meljanac-3} based on the results in \cite{Meljanac-6}. The additional generators $\hat{x}_{\mu\nu}$ are assumed to satisfy the commutation relations
\begin{align}
&\big[ \hat{x}_{\mu\nu},\hat{x}_{\rho\sigma}\big]={\rm i}\big( \eta_{\mu\rho} \hat{x}_{\nu\sigma} - \eta_{\mu\sigma} \hat{x}_{\nu\rho} - \eta_{\nu\rho} \hat{x}_{\mu\sigma} + \eta_{\nu\sigma} \hat{x}_{\mu\rho}\big), \label{3}\\
&\big[ \hat{x}_{\mu\nu},x_{\lambda}\big]=0,\qquad 
\big[ \hat{x}_{\mu\nu},p_{\lambda}\big]=0. \label{3.2}
\end{align}
In this case, we consider realizations of the Lorentz generators of the form
\begin{align*}
&M_{\mu\nu}=\hat{x}_{\mu\nu}+x_{\mu}p_{\nu}-x_{\nu}p_{\mu},\qquad M_{\mu\nu} \rhd 1=\hat{x}_{\mu\nu} \rhd 1=x_{\mu\nu} \qquad \text{and} \qquad p_{\mu}\rhd 1=0,
\end{align*}
where $x_{\mu\nu}$ are commuting variables.
\begin{Theorem}\label{2-T}
Extension of the Snyder realization \eqref{2.1}--\eqref{2.2} with additional generators $\hat{x}_{\mu\nu}$ is given by
\begin{align}
&\hat{x}_{\mu}=x_{\mu}+\beta^{2}(x\cdot p)p_{\mu}-\beta^{2}\hat{x}_{\mu\alpha}p_{\alpha}\frac{1}{1+\sqrt{1+u}},\label{3.5}\\
&M_{\mu\nu}=\hat{x}_{\mu\nu}+x_{\mu}p_{\nu}-x_{\nu}p_{\mu}.\label{3.6}
\end{align}
\end{Theorem}
\begin{proof}
In order to prove that we can construct the realization of the Snyder model by \eqref{3.5} and \eqref{3.6}, we show that \eqref{3.5} and \eqref{3.6} satisfy Snyder algebra \eqref{1.1}--\eqref{1.3}. A short computation using \eqref{2.6} yields
{\samepage\begin{align}
&\left[ x_{\mu}, \frac{1}{1+\sqrt{1+u}}\right] =\frac{-{\rm i}\beta^{2}p_{\mu}}{\sqrt{1+u}\big( 1+\sqrt{1+u}\big) ^{2}},\label{3.7}\\
&\left[ \frac{1}{1+\sqrt{1+u}}, (x\cdot p)\right] =\frac{{\rm i}u}{\sqrt{1+u}\big( 1+\sqrt{1+u}\big)^{2}},\label{3.8}
\end{align}}
and
\begin{equation}\label{3.9}
\left[ M_{\mu\nu},\frac{1}{1+\sqrt{1+u}} \right]=0.
\end{equation}
Now, from \eqref{3}--\eqref{3.2} and \eqref{3.7}--\eqref{3.8} using bilinearity of the Lie bracket, we obtain
\begin{align*}
\big[ \hat{x}_{\mu}, \hat{x}_{\nu}\big] &=\left[ x_{\mu}+\beta^{2}(x\!\cdot \! p)p_{\mu}-\beta^{2}\hat{x}_{\mu\alpha}p_{\alpha}\frac{1}{1+\!\sqrt{1+u}},x_{\nu}+\beta^{2}(x\!\cdot \! p)p_{\nu}-\beta^{2}\hat{x}_{\nu\rho}p_{\rho}\frac{1}{1+\!\sqrt{1+u}}\right] \\
&={\rm i}\beta^{2}(x_{\mu}p_{\nu}-x_{\nu}p_{\mu})+{\rm i}2\beta^{2}\frac{\hat{x}_{\mu\nu}}{1+\sqrt{1+u}}+{\rm i}\beta^{2}u\frac{\hat{x}_{\mu\nu}}{( 1+\sqrt{1+u})^{2} }\\
&={\rm i}\beta^{2}\big(x_{\mu}p_{\nu}-x_{\nu}p_{\mu}+\hat{x}_{\mu\nu}\big)\\
&={\rm i}\beta^{2}M_{\mu\nu}.
\end{align*}
Similarly, by using \eqref{3.9}, we check that
\eqref{3.5} and \eqref{3.6} satisfy \eqref{1.2}--\eqref{1.3}, therefore
\eqref{3.5} and~\eqref{3.6} is a realization of the extended Snyder model.
\end{proof}

In order to obtain a family of realizations of the extended Snyder model, we use similarity transformations from Section~\ref{sec-2}, by $S={\rm e}^{{\rm i}G}$ where $G=(x\cdot p)F(u)$. First, note that
\begin{align}
S\left( \frac{1}{1+\sqrt{1+u}}\right) S^{-1}&=S\left( \sum_{m=1}^{\infty}\binom{\frac{1}{2}}{m}u^{m-1}\right) S^{-1} \notag\\
&=\sum_{m=1}^{\infty}\binom{\frac{1}{2}}{m}\big(\beta^{2}p'^{2}\big)^{m-1}=\frac{1}{1+\sqrt{1+\beta^{2}p'^{2}}}\label{3.10}
\end{align}
and
\begin{equation}\label{3.11}
S\big( \hat{x}_{\mu\nu}\big) S^{-1}=\hat{x}_{\mu\nu}.
\end{equation}
Now \eqref{3.10} and \eqref{3.11} implies that
\begin{align*}
\hat{x}_{\mu}&=S\left( x_{\mu}+\beta^{2}(x\cdot p)p_{\mu}-\beta^{2}\hat{x}_{\mu\alpha}p_{\alpha}\frac{1}{1+\sqrt{1+u}}\right) S^{-1}\\
&=x_{\mu}'+\beta^{2}(x'\cdot p')p_{\mu}'-\beta^{2}\hat{x}_{\mu\alpha}p_{\alpha}'\frac{1}{1+\sqrt{1+\beta^{2}p'^{2}}}.
\end{align*}
Finally, by using results given in Section~\ref{sec-2}, \eqref{2.22} and \eqref{2.25}, we obtain a family of realizations of the extended Snyder model
\begin{equation}\label{3.12}
\hat{x}_{\mu}=x_{\mu}\varphi_{1}(u)+\beta^{2}(x\cdot p)p_{\mu}\varphi_{2}(u)-\beta^{2}\hat{x}_{\mu\alpha}p_{\alpha}\frac{1}{\varphi_{1}(u)+\sqrt{\varphi_{1}^{2}(u)+u}},
\end{equation}
where $\varphi_{1}(u)$ and $\varphi_{2}(u)$ satisfy \eqref{2.4}.

Note that realizations \eqref{3.5}, \eqref{3.6}, \eqref{3.12} and \eqref{2.4} are the exact results written in closed form.

\section[k-deformed extended Snyder model]{$\boldsymbol{\kappa}$-deformed extended Snyder model}\label{sec-4}

In this section we consider a family of Lie algebras containing $\kappa$-Poincar\'{e} \cite{Borowiec, Kowalski-2, Lukierski-4, Lukierski-3, Meljanac-12, Meljanac-11} and Snyder algebras as special cases. They are generated by the NC coordinates $\hat{x}_{\mu}$ and Lorentz generators $M_{\mu\nu}$ satisfying
\begin{align}
&\big[ \hat{x}_{\mu},\hat{x}_{\nu}\big] ={\rm i}\big( a_{\mu}\hat{x}_{\nu}-a_{\nu}\hat{x}_{\mu}+\beta^{2}M_{\mu\nu}\big),\label{4.1}\\
&\big[ M_{\mu\nu},\hat{x}_{\lambda}\big] =-{\rm i}\big( \hat{x}_{\mu}\eta_{\nu\lambda}-
\hat{x}_{\nu}\eta_{\mu\lambda}+a_{\mu}M_{\nu\lambda}-a_{\nu}M_{\mu\lambda}\big),\label{4.2}\\
&[ M_{\mu\nu},M_{\rho\sigma}] ={\rm i}( \eta_{\mu\rho} M_{\nu\sigma} - \eta_{\mu\sigma} M_{\nu\rho} - \eta_{\nu\rho} M_{\mu\sigma} + \eta_{\nu\sigma} M_{\mu\rho}),\label{4.3}
\end{align}
where $a_{\mu}=\frac{1}{\kappa}u_{\mu}$, $u^{2}=(-1,0,1)$ and $\kappa$ is the mass parameter with $\frac{1}{\kappa}\neq \beta$. Such models were considered in \cite{Meljanac-7, Meljanac-8} and the $\kappa$-deformed extended Snyder model was considered in \cite{Lukierski-2, Meljanac-12, Meljanac-9}.

If $M_{\mu\nu}=x_{\mu}p_{\nu}-x_{\nu}p_{\mu}$ and $[ M_{\mu\nu}, p_{\lambda}] ={\rm i}( p_{\nu}\eta_{\mu\lambda}-p_{\mu}\eta_{\nu\lambda}),$ then one particular realization of above algebra is given in \cite{Meljanac-7, Meljanac-8} with
\begin{equation*}
\hat{x}_{\mu}=x_{\mu}\sqrt{1+\big( a^{2}-\beta^{2}\big)p^{2}}+M_{\mu\alpha}a_{\alpha}.
\end{equation*}
For $a_{\mu}=0$, we get a realization of the Snyder model
\begin{equation*}
\hat{x}_{\mu}=x_{\mu}\sqrt{1-u}.
\end{equation*}
For $\beta^{2}=0$, we get the natural realization \cite{Meljanac-12, Meljanac-11}, i.e.,
a realization in the classical basis \cite{Borowiec} of the $\kappa$-Poincar\'{e} algebra
\begin{equation*}
\hat{x}_{\mu}=x_{\mu}\sqrt{1+a^{2}p^{2}}+M_{\mu\alpha}a_{\alpha}.
\end{equation*}
In the following paper we present the exact new result for the $\kappa$-deformed extended Snyder model that is written in closed form and different from the perturbative results discussed in~\cite{Meljanac-10, Meljanac-9}.

\begin{Theorem}\label{3-T}
	Let
	\begin{equation}\label{4.11}
	M_{\mu\nu}=\hat{x}_{\mu\nu}+x_{\mu}p_{\nu}-x_{\nu}p_{\mu}.
	\end{equation}
	Then one particular realization of the algebra \eqref{4.1}--\eqref{4.3} is given by
	\begin{equation}\label{4.4} \hat{x}_{\mu}=x_{\mu}\sqrt{1+\big(a^{2}-\beta^{2}\big)p^{2}}+M_{\mu\alpha}a_{\alpha}+\big(a^{2}-
\beta^{2}\big)\hat{x}_{\mu\alpha}p_{\alpha}\frac{1}{1+\sqrt{1+\big( a^{2}-\beta^{2}\big)p^{2}}}.
	\end{equation}
\end{Theorem}
\begin{proof} We have to show that realization \eqref{4.4} satisfies the algebra \eqref{4.1}--\eqref{4.3}. By using \eqref{3}--\eqref{3.2}, it is easy to see that
\begin{align}
&[ M_{\mu\nu}, p_{\lambda}]={\rm i}(p_{\nu}\eta_{\mu\lambda}- p_{\mu}\eta_{\nu\lambda}),\qquad
[ M_{\mu\nu}, x_{\lambda}] ={\rm i}( x_{\nu}\eta_{\mu\lambda}-x_{\mu}\eta_{\nu\lambda})\label{4.5}
\end{align}
and
\begin{align}\label{4.7}
\big[M_{\mu\nu},\hat{x}_{\rho\sigma}\big]={\rm i}\big( \eta_{\mu\rho} \hat{x}_{\nu\sigma} - \eta_{\mu\sigma} \hat{x}_{\nu\rho} - \eta_{\nu\rho} \hat{x}_{\mu\sigma} + \eta_{\nu\sigma} \hat{x}_{\mu\rho}\big).
\end{align}
Furthermore, from \eqref{2.6} we get
\begin{align}\label{4.8}
&\left[ x_{\mu}, \frac{1}{1+\sqrt{1+\big( a^{2}-\beta^{2}\big) p^{2}}}\right] =\frac{-{\rm i}\big( a^{2}-\beta^{2}\big) p_{\mu}}{\sqrt{1+\big( a^{2}-\beta^{2}\big)p^{2}}\Big( 1+\sqrt{1+\big( a^{2}-\beta^{2}\big)p^{2}}\Big) ^{2}},\\
&\left[ x_{\mu}, \sqrt{1+\big( a^{2}-\beta^{2}\big)p^{2}}\right] =\frac{{\rm i}\big( a^{2}-\beta^{2}\big)p^{2}}{\sqrt{1+\big( a^{2}-\beta^{2}\big)p^{2}}},\label{4.9}
\end{align}
and
\begin{equation}\label{4.10}
\left[ M_{\mu\nu}, \frac{1}{1+\sqrt{1+\big( a^{2}-\beta^{2}\big)p^{2}}} \right] =\left[M_{\mu\nu}, \sqrt{1+\big( a^{2}-\beta^{2}\big)p^{2}}\right]=0.
\end{equation}
{\samepage Now, from \eqref{4.5}--\eqref{4.10} we have
\begin{align*}
\big[\hat{x}_{\mu},\hat{x}_{\nu}\big]={}&\Big[ x_{\mu}\sqrt{1+\big( a^{2}-\beta^{2}\big)p^{2}}+M_{\mu\alpha}a_{\alpha}+\big( a^{2}-\beta^{2}\big)\hat{x}_{\mu\alpha}p_{\alpha}\frac{1}{1+\sqrt{1+\big( a^{2}-\beta^{2}\big)p^{2}}},\\
&\quad x_{\nu}\sqrt{1+\big( a^{2}-\beta^{2}\big)p^{2}}+M_{\nu\rho}a_{\rho}+\big( a^{2}-\beta^{2}\big)\hat{x}_{\nu\rho}p_{\rho}\frac{1}{1+\sqrt{1+\big( a^{2}-\beta^{2}\big)p^{2}}}\Big]\\
={}&-{\rm i}\big( a^{2}-\beta^{2}\big)(x_{\mu}p_{\nu}-x_{\nu}p_{\mu})+{\rm i}(a_{\mu}x_{\nu}-a_{\nu}x_{\mu})\sqrt{1+\big( a^{2}-\beta^{2}\big)p^{2}}\\
&+{\rm i}(a_{\mu}M_{\nu\rho}a_{\rho}-a_{\nu}M_{\mu\alpha}a_{\alpha}+a^{2}M_{\mu\nu})-{\rm i}\big( a^{2}-\beta^{2}\big)\hat{x}_{\mu\nu}\\
={}&{\rm i}\big( a_{\mu}\hat{x}_{\nu}-a_{\nu}\hat{x}_{\mu}+\beta^{2}M_{\mu\nu}\big).
\end{align*}}
In similar way, by using \eqref{4.5}--\eqref{4.10}, we show that \eqref{4.11} and \eqref{4.4} satisfy \eqref{4.2}--\eqref{4.3}.
\end{proof}

For $a_{\mu}=0$, we get the realization of the extended Snyder model found in Section~\ref{sec 3}
\begin{equation*}
\hat{x}_{\mu}=x_{\mu}\sqrt{1-u}-\beta^{2}\hat{x}_{\mu\alpha}p_{\alpha}\frac{1}{1+\sqrt{1-u}}.
\end{equation*}
For $\beta^{2}=0$, we find
\begin{equation*}
\hat{x}_{\mu}=x_{\mu}\sqrt{1+a^{2}p^{2}}+M_{\mu\alpha}a_{\alpha}+a^{2}\hat{x}_{\mu\alpha}p_{\alpha}\frac{1}{1+\sqrt{1+a^{2}p^{2}}}.
\end{equation*}
This is a new result corresponding to the $\kappa$-Poincar\'{e} algebra with additional tensorial generators~$\hat{x}_{\mu\nu}$. The most general realizations of $\hat{x}_{\mu}$ in all cases in this section are obtained by using the most general corresponding similarity transformations.
Construction of Hermitian realizations in Sections~\ref{sec 3} and~\ref{sec-4} can be obtained simply by changing
$\hat{x}_{\mu}$ with $\frac{1}{2} \big(\hat{x}_{\mu}+ \hat{x}^{\dagger}_{\mu}\big)$, as in Remark~\ref{rem-3}.

\section{Conclusion and discussion}\label{sec-5}

In Section \ref{sec-2}, we defined similarity transformations \eqref{2.3} and using Propositions~\ref{1-P} and~\ref{2-P}, we proved realizations of the Snyder model~\eqref{2.4} in Theorem~\ref{1-T}. This result was obtained in~\cite{Battisti-2, Battisti-1} without using similarity transformations. In Section~\ref{sec 3}, we gave a proof of Theorem~\ref{2-T} (equations~\eqref{3.5}--\eqref{3.6}) that includes additional tensorial generators $\hat x_{\mu \nu}$ and it is a generalization of the original Snyder realization. This is a new exact result leading to an associative star product and coassociative coproduct \cite{Meljanac-1}. Also, we obtained exact results for the realizations of the extended Snyder model with functions $\varphi_{1}(u)$ and $\varphi_{2}(u)$ \eqref{3.12} using Propositions~\ref{1-P} and~\ref{2-P}. In Section~\ref{sec-4}, we proved Theorem~\ref{3-T} (equations~\eqref{4.11} and \eqref{4.4})
and this is a new exact result for the $\kappa$-deformed extended Snyder model.

The physical role of the additional tensorial generators $\hat x_{\mu \nu}$ is not completely clear, except that they mathematically lead to an associative star product and coassociative coproduct \cite{Girelli, Meljanac-1}. Some attempts for applications of the extended Snyder model were made in \cite{Girelli, Lukierski-1, Meljanac-1} and of the $\kappa$-deformed extended Snyder model in \cite{Lukierski-2, Meljanac-10,Meljanac-9}. Possible applications of the generalizations of the Snyder model to curved spaces were discussed in \cite{Meljanac-4, Meljanac-5}. The future prospect of our investigation is the construction of the star product and twist.

\subsection*{Acknowledgement}
We thank S.~Kre\v{s}i\'c Juri\'c, Z.~\v{S}koda and S.~Mignemi for useful comments.

\pdfbookmark[1]{References}{ref}
\LastPageEnding


\begin{thebibliography}{99}
\footnotesize\itemsep=0pt

\bibitem{Amelino-2}
Amelino-Camelia G., Arzano M., Coproduct and star product in field theories on
 {L}ie-algebra noncommutative space-times, \href{https://doi.org/10.1103/PhysRevD.65.084044}{\textit{Phys. Rev.~D}} \textbf{65}
 (2002), 084044, 8~pages, \href{https://arxiv.org/abs/hep-th/0105120}{arXiv:hep-th/0105120}.

\bibitem{Amelino-Camelia-1}
Amelino-Camelia G., Lukierski J., Nowicki A., {$\kappa$}-deformed covariant
 phase space and quantum-gravity uncertainty relations, \textit{Phys. Atomic
 Nuclei} \textbf{61} (1998), 1811--1815, \href{https://arxiv.org/abs/hep-th/9706031}{arXiv:hep-th/9706031}.

\bibitem{Arzano}
Arzano M., Kowalski-Glikman J., Deformations of spacetime symmetries~--
 gravity, group-valued momenta, and non-commutative fields, \textit{Lecture
 Notes in Phys.}, Vol.~986, \href{https://doi.org/10.1007/978-3-662-63097-6}{Springer}, Berlin, 2021.

\bibitem{Banerjee}
Banerjee R., Kumar K., Roychowdhury D., Symmetries of {S}nyder--de {S}itter
 space and relativistic particle dynamics, \href{https://doi.org/10.1007/JHEP03(2011)060}{\textit{J.~High Energy Phys.}}
 \textbf{2011} (2011), no.~3, 060, 14~pages, \href{https://arxiv.org/abs/1101.2021}{arXiv:1101.2021}.

\bibitem{Battisti-2}
Battisti M.V., Meljanac S., Modification of {H}eisenberg uncertainty relations
 in noncommutative {S}nyder space-time geometry, \href{https://doi.org/10.1103/PhysRevD.79.067505}{\textit{Phys. Rev.~D}}
 \textbf{79} (2009), 067505, 4~pages, \href{https://arxiv.org/abs/0812.3755}{arXiv:0812.3755}.

\bibitem{Battisti-1}
Battisti M.V., Meljanac S., Scalar field theory on noncommutative {S}nyder
 spacetime, \href{https://doi.org/10.1103/PhysRevD.82.024028}{\textit{Phys. Rev.~D}} \textbf{82} (2010), 024028, 9~pages,
 \href{https://arxiv.org/abs/1003.2108}{arXiv:1003.2108}.

\bibitem{Borowiec}
Borowiec A., Pacho{\l} A., The classical basis for the
 {$\kappa$}-{P}oincar\'{e} {H}opf algebra and doubly special relativity
 theories, \href{https://doi.org/10.1088/1751-8113/43/4/045203}{\textit{J.~Phys.~A}} \textbf{43} (2010), 045203, 10~pages,
 \href{https://arxiv.org/abs/0903.5251}{arXiv:0903.5251}.

\bibitem{Doplicher}
Doplicher S., Fredenhagen K., Roberts J.E., Spacetime quantization induced by
 classical gravity, \href{https://doi.org/10.1016/0370-2693(94)90940-7}{\textit{Phys. Lett.~B}} \textbf{331} (1994), 39--44.

\bibitem{Doplicher-2}
Doplicher S., Fredenhagen K., Roberts J.E., The quantum structure of spacetime
 at the {P}lanck scale and quantum fields, \href{https://doi.org/10.1007/BF02104515}{\textit{Comm. Math. Phys.}}
 \textbf{172} (1995), 187--220, \href{https://arxiv.org/abs/hep-th/0303037}{arXiv:hep-th/0303037}.

\bibitem{Franchino}
Franchino-Vi\~{n}as S.A., Mignemi S., Worldline formalism in {S}nyder spaces,
 \href{https://doi.org/10.1103/physrevd.98.065010}{\textit{Phys. Rev.~D}} \textbf{98} (2018), 065010, 10~pages,
 \href{https://arxiv.org/abs/1806.11467}{arXiv:1806.11467}.

\bibitem{Garay}
Garay L.J., Quantum gravity and minimum length, \href{https://doi.org/10.1142/S0217751X95000085}{\textit{Internat.~J. Modern
 Phys.~A}} \textbf{10} (1995), 145--165, \href{https://arxiv.org/abs/gr-qc/9403008}{arXiv:gr-qc/9403008}.

\bibitem{Girelli}
Girelli F., Livine E.R., Scalar field theory in {S}nyder space-time:
 alternatives, \href{https://doi.org/10.1007/JHEP03(2011)132}{\textit{J.~High Energy Phys.}} \textbf{2011} (2011), no.~3, 132,
 31~pages, \href{https://arxiv.org/abs/1004.0621}{arXiv:1004.0621}.

\bibitem{Guo}
Guo H.-Y., Huang C.-G., Wu H.-T., Yang's model as triply special relativity and
 the {S}nyder's model -- de {S}itter special relativity duality, \href{https://doi.org/10.1016/j.physletb.2008.04.012}{\textit{Phys.
 Lett.~B}} \textbf{663} (2008), 270--274, \href{https://arxiv.org/abs/0801.1146}{arXiv:0801.1146}.

\bibitem{Kowalski-2}
Kowalski-Glikman J., Nowak S., Non-commutative space-time of doubly special
 relativity theories, \href{https://doi.org/10.1142/S0218271803003050}{\textit{Internat.~J. Modern Phys.~D}} \textbf{12} (2003),
 299--315, \href{https://arxiv.org/abs/hep-th/0204245}{arXiv:hep-th/0204245}.

\bibitem{Kowalski}
Kowalski-Glikman J., Smolin L., Triply special relativity, \href{https://doi.org/10.1103/PhysRevD.70.065020}{\textit{Phys.
 Rev.~D}} \textbf{70} (2004), 065020, 6~pages, \href{https://arxiv.org/abs/hep-th/0406276}{arXiv:hep-th/0406276}.

\bibitem{Lukierski-1}
Lukierski J., Meljanac S., Mignemi S., Pacho{\l} A., Quantum perturbative
 solutions of extended {S}nyder and {Y}ang models with spontaneous symmetry
 breaking, \href{https://arxiv.org/abs/2212.02316}{arXiv:2212.02316}.

\bibitem{Lukierski-2}
Lukierski J., Meljanac S., Mignemi S., Pacho{\l} A., Generalized quantum phase
 spaces for the {$\kappa$}-deformed extended {S}nyder model, \href{https://doi.org/10.1016/j.physletb.2023.137709}{\textit{Phys.
 Lett.~B}} \textbf{838} (2023), 137709, 8~pages, \href{https://arxiv.org/abs/2208.06712}{arXiv:2208.06712}.

\bibitem{Lukierski-4}
Lukierski J., Nowicki A., Ruegg H., New quantum {P}oincar\'{e} algebra and
 {$\kappa$}-deformed field theory, \href{https://doi.org/10.1016/0370-2693(92)90894-A}{\textit{Phys. Lett.~B}} \textbf{293} (1992),
 344--352.

\bibitem{Lukierski-3}
Lukierski J., Ruegg H., Nowicki A., Tolstoy V.N., {$q$}-deformation of
 {P}oincar\'{e} algebra, \href{https://doi.org/10.1016/0370-2693(91)90358-W}{\textit{Phys. Lett.~B}} \textbf{264} (1991), 331--338.

\bibitem{Majid}
Majid S., Foundations of quantum group theory, \href{https://doi.org/10.1017/CBO9780511613104}{Cambridge University Press},
 Cambridge, 1995.

\bibitem{Meljanac-12}
Meljanac S., Kre\v{s}i\'c-Juri\'c S., Differential structure on
 {$\kappa$}-{M}inkowski space, and {$\kappa$}-{P}oincar\'{e} algebra,
 \href{https://doi.org/10.1142/S0217751X11053948}{\textit{Internat.~J. Modern Phys.~A}} \textbf{26} (2011), 3385--3402,
 \href{https://arxiv.org/abs/1004.4647}{arXiv:1004.4647}.

\bibitem{Meljanac-11}
Meljanac S., Kre\v{s}i\'c-Juri\'c S., Stoji\'c M., Covariant realizations
 of kappa-deformed space, \href{https://doi.org/10.1140/epjc/s10052-007-0285-8}{\textit{Eur. Phys.~J.~C}} \textbf{51} (2007),
 229--240, \href{https://arxiv.org/abs/hep-th/0702215}{arXiv:hep-th/0702215}.

\bibitem{Meljanac-6}
Meljanac S., Martini\'c~Bila\'c T., Kre\v{s}i\'c-Juri\'c S.,
 Generalized {H}eisenberg algebra applied to realizations of the orthogonal,
 {L}orentz, and {P}oincar\'{e} algebras and their dual extensions,
 \href{https://doi.org/10.1063/5.0006184}{\textit{J.~Math. Phys.}} \textbf{61} (2020), 051705, 13~pages,
 \href{https://arxiv.org/abs/2003.02726}{arXiv:2003.02726}.

\bibitem{Meljanac-7}
Meljanac S., Meljanac D., Samsarov A., Stoji\'c M., {$\kappa$}-deformed
 {S}nyder spacetime, \href{https://doi.org/10.1142/S0217732310032652}{\textit{Modern Phys. Lett.~A}} \textbf{25} (2010),
 579--590, \href{https://arxiv.org/abs/0912.5087}{arXiv:0912.5087}.

\bibitem{Meljanac-8}
Meljanac S., Meljanac D., Samsarov A., Stojic M., Kappa {S}nyder deformations
 of {M}inkowski spacetime, realizations and {H}opf algebra, \href{10.1103/PhysRevD.83.065009}{\textit{Phys.
 Rev.~D}} \textbf{83} (2011), 065009, 16~pages, \href{https://arxiv.org/abs/1102.1655}{arXiv:1102.1655}.

\bibitem{Meljanac-1}
Meljanac S., Mignemi S., Associative realizations of the extended {S}nyder
 model, \href{https://doi.org/10.1103/physrevd.102.126011}{\textit{Phys. Rev.~D}} \textbf{102} (2020), 126011, 9~pages,
 \href{https://arxiv.org/abs/2007.13498}{arXiv:2007.13498}.

\bibitem{Meljanac-10}
Meljanac S., Mignemi S., Associative realizations of {$\kappa$}-deformed
 extended {S}nyder model, \href{https://doi.org/10.1103/physrevd.104.086006}{\textit{Phys. Rev.~D}} \textbf{104} (2021), 086006,
 9~pages, \href{https://arxiv.org/abs/2106.08131}{arXiv:2106.08131}.

\bibitem{Meljanac-9}
Meljanac S., Mignemi S., Unification of {$\kappa$}-{M}inkowski and extended
 {S}nyder spaces, \href{https://doi.org/10.1016/j.physletb.2021.136117}{\textit{Phys. Lett.~B}} \textbf{814} (2021), 136117, 5~pages,
 \href{https://arxiv.org/abs/2101.05275}{arXiv:2101.05275}.

\bibitem{Meljanac-4}
Meljanac S., Mignemi S., Generalizations of {S}nyder model to curved spaces,
 \href{https://doi.org/10.1016/j.physletb.2022.137289}{\textit{Phys. Lett.~B}} \textbf{833} (2022), 137289, 6~pages,
 \href{https://arxiv.org/abs/2206.04772}{arXiv:2206.04772}.

\bibitem{Meljanac-5}
Meljanac S., Mignemi S., Noncommutative {Y}ang model and its generalizations,
 \href{https://doi.org/10.1063/5.0135492}{\textit{J.~Math. Phys.}} \textbf{64} (2023), 023505, 9~pages,
 \href{https://arxiv.org/abs/2211.11755}{arXiv:2211.11755}.

\bibitem{Meljanac-2}
Meljanac S., Pacho{\l} A., Heisenberg doubles for {S}nyder-type models,
 \href{https://doi.org/10.3390/sym13061055}{\textit{Symmetry}} \textbf{13} (2021), 1055, 13~pages, \href{https://arxiv.org/abs/2101.02512}{arXiv:2101.02512}.

\bibitem{Meljanac-3}
Meljanac S., \v{S}trajn R., Deformed quantum phase spaces, realizations, star
 products and twists, \href{https://doi.org/10.3842/SIGMA.2022.022}{\textit{SIGMA}} \textbf{18} (2022), 022, 20~pages,
 \href{https://arxiv.org/abs/2112.12038}{arXiv:2112.12038}.

\bibitem{Mignemi-3}
Mignemi S., The {S}nyder--de {S}itter model from six dimensions,
 \href{https://doi.org/10.1088/0264-9381/26/24/245020}{\textit{Classical Quantum Gravity}} \textbf{26} (2009), 245020, 9~pages.

\bibitem{Mignemi-2}
Mignemi S., Rosati G., Relative-locality phenomenology on {S}nyder spacetime,
 \href{https://doi.org/10.1088/1361-6382/aac9d5}{\textit{Classical Quantum Gravity}} \textbf{35} (2018), 145006, 17~pages,
 \href{https://arxiv.org/abs/1803.02134}{arXiv:1803.02134}.

\bibitem{Mignemi-1}
Mignemi S., \v{S}trajn R., Snyder dynamics in a {S}chwarzschild spacetime,
 \href{https://doi.org/10.1103/PhysRevD.90.044019}{\textit{Phys. Rev.~D}} \textbf{90} (2014), 044019, 5~pages,
 \href{https://arxiv.org/abs/1404.6396}{arXiv:1404.6396}.

\bibitem{Snyder}
Snyder H.S., Quantized space-time, \href{https://doi.org/10.1103/PhysRev.71.38}{\textit{Phys. Rev.}} \textbf{71} (1947),
 38--41.

\end{thebibliography}
\end{document}